\documentclass[a4paper]{article}
\usepackage{braket}
\usepackage{amssymb}
\usepackage{amsmath}
\usepackage{mathtools}
\usepackage{url}
\usepackage{amsthm}
\newtheorem{theorem}{Theorem}

\usepackage{vub}
\DeclareMathOperator{\Tr}{Tr}

\title{To wedge or not to wedge}
\pretitle{\flushleft{Thesis submitted in partial fulfilment of the requirements for the bachelor's degree of Physics and Astronomy}}
\subtitle{Wedges and operator reconstructability in toy models of AdS/CFT}
\author{Vic Vander Linden}
\date{June~2023}
\promotors{Promotors: Prof.\ Dr.\ Ben Craps \and  Dr.\ Juan Hernandez \and \ Philip Hacker
}
\faculty{sciences and bioengineering sciences}

\begin{document}
\maketitle%

\section*{Abstract}
The AdS/CFT correspondence is an explicit realization of the holographic principle relating a theory of gravity in a volume of space to a lower dimensional quantum field theory on its boundary. By exploiting elements of quantum error correction, qubit toy models of this correspondence have been constructed for which the bulk logical operators are representable by operators acting on the boundary \cite{Happy}. Given a boundary subregion, wedges in the volume space are used to enclose the bulk qubits for which logical operators are reconstructable on that boundary subregion. In this thesis a number of different wedges, such as the causal wedge, greedy entanglement wedge and minimum entanglement wedge, are examined. More specifically, Monte-Carlo simulations of boundary erasure are performed with various toy models to study the differences between wedges and the effect on these wedge by the type of the model, non-uniform boundaries and stacking of models. It has been found that the minimum entanglement wedge is the best approximate for the true geometric wedge. This is illustrated by an example toy model for which an operator beyond the greedy entanglement wedge was also reconstructed. In addition, by calculating the entropy of these subregions, the viability of a mutual information wedge is rejected. Only for particular connected boundary subregions was the inclusion of the central tensor by the geometric wedge associated to a rise in mutual information.
\newpage

\tableofcontents%

\newpage 
\section{Introduction}%
In 1981, Richard Feynman jokingly called to close the doors in an attempt to trap his attendees at the first Physics and Computation conference. He did so as he revealed the only way to truly simulate physics might be through quantum mechanics \cite{Feynman}. His ideas about an inherent probabilistic computer helped to form and popularize a new branch of physics, quantum information. \\ \\
In this thesis, a brief introduction to quantum information, quantum error correction, tensors and tensor networks and the AdS/CFT correspondence will be given. For this purpose, a literary study of the following works was conducted \cite{Happy} - \cite{google}. Thereafter, wedges and operator reconstructability for qubit toy models of this AdS/CFT correspondence are investigated using Monte-Carlo simulations of central qubit inclusion with boundary erasure, inspired by the similar simulations performed by Pastawski et al. in \cite{Happy}. These original simulations are recreated in Section \ref{sectionwed}. In addition, further Monte-Carlo simulations are carried out to investigate other properties of tensor networks such as a different shape of the network and the effect of a non-uniform boundary. Additionally, the effect of stacking tensor networks, as introduced by \cite{stack}, is investigated using the same simulations. Furthermore, an additional wedge will be proposed, the minimal entanglement wedge. The Monte-Carlo simulations using this wedge will depend on the state of the bulk qubits. Its particular behaviour for single states will be explained using a theorem constructed and proven in this thesis. This minimal entanglement wedge will be shown to be useful in approximating the true geometric entanglement wedge of toy tensor networks. Finally, the relationship between these wedges and operator reconstructability of bulk qubits will be investigated using newly constructed examples.\\ \\
In the third chapter of this thesis the viability of a mutual information wedge will be determined. The connection between inclusion in the geometric wedge and the rise in mutual information will be investigated for both connected and non connected boundary subregions. Additionally, the effect of different bulk states, tensor networks and layer depth on the mutual information will be examined. Furthermore, the computation of both the Von Neumann- and Ryu-Takayanagi entropy will be compared for a growing boundary subregion of one explicit example.
\subsection{Qubits}
\label{qubit}
The foundation of modern quantum information theory lies in the qubit. This term, first coined by Schumacher in 1995 \cite{qubit}, represents a system in a superposition of two possible energy levels which are denoted by $\ket{0}$ and $\ket{1}$,
\begin{equation}
\ket{\Psi} = \alpha \ket{0} + \beta \ket{1} = \alpha \begin{pmatrix} 1 \\ 0 \end{pmatrix} + \beta \begin{pmatrix} 0 \\ 1 \end{pmatrix},
\label{general}
\end{equation}
where $\ket{\Psi}$ represents the state of the system and $\alpha, \beta \in \mathbb{C}$ with $|\alpha|^2 + |\beta|^2 = 1$. \\ \\
These qubits are isomorphic to spin-$\frac{1}{2}$ systems, spanning a two dimensional Hilbert space $\mathcal{H}_2$. In recent quantum computers, such as Google's 53 qubit and IBM's 433 qubit quantum processor, qubits are commonly approximated by letting current flow through a superconducting material \cite{Google, IBM}.  Qubits have a property that regular bits do not possess, entanglement. As an example, consider an Einstein–Podolsky–Rosen (EPR) pair or Bell state of two qubits \cite{EPR},
\begin{equation}
    \ket{\Psi} = \frac{\ket{00} + \ket{11}}{\sqrt{2}},
\end{equation}
where $\ket{00}$ and $\ket{11}$ are the tensor products $\ket{0} \otimes \ket{0}$ and $\ket{1} \otimes \ket{1}$ respectively. \\ \\
Suppose one wishes to describe only one of these qubits independently. This can be done by \textit{tracing out} the other qubit. This mathematical procedure is conducted by performing a partial trace of the density matrix $\rho$. This density matrix is defined for a quantum system as $ \rho = \sum_{j}{p_j \ket{\Psi_j}\bra{\Psi_j}}$, where $j$ sums over a basis of pure states for the system with their probability $p_j$. As an example the density matrix of the EPR pair is represented by,
\begin{equation}
 \rho = 1 \cdot  \left(\frac{\ket{00} + \ket{11}}{\sqrt{2}}\right) \left(\frac{\bra{00} + \bra{11}}{\sqrt{2}}\right) = \begin{pmatrix}
1/2 & 0 & 0 & 1/2 \\
0 & 0 & 0 & 0 \\
0 & 0 & 0 & 0 \\
1/2 & 0 & 0 & 1/2 \\
\end{pmatrix}.
\end{equation}
where the matrix is the result of the tensor product of the individual Hilbert spaces. This density matrix represents a pure state because for the normalized density matrix the trace of the square, $\Tr(\rho^2)$, is equal to one \cite{griffiths}. Tracing out a subsystem $A$ of a density matrix then involves taking the partial trace of subsystem B defined as \cite{trace},
\begin{equation}
\Tr_B(\rho_{AB}) = \sum_j{(I_A \otimes \bra{j}_B) \rho_{AB} (I_A \otimes \ket{j}_B)}
\end{equation}
where $\set{\ket{j}_B}$ is any orthonormal basis for the Hilbert space $\mathcal{H}_B$ of subsystem B and $I_A$ is the identity operator of the remaining Hilbert space $\mathcal{H}_A$. As an example, tracing out one of the qubits of the EPR pair, will yield for the remaining qubit the mixed state with density matrix,
 \begin{equation}
     \rho_A = \Tr_B(\rho_{AB}) = \frac{1}{2}(\ket{0_A}\bra{0_A} + \ket{1_A}\bra{1_A}) = \begin{pmatrix}
         1/2 & 0 \\
         0 & 1/2 \\
     \end{pmatrix}.\label{mix}
 \end{equation}
This procedure will not result a pure linear superposition as defined in Formula (\ref{general}), but a mixed state made up of an ensemble of the states $ \set{\bra{0_A},\bra{1_A}}$ with a probability of $\frac{1}{2}$ each. Since it is a mixed state, $\Tr(\rho^2)$ will no longer be equal to one. \\ \\ In addition, calculating the Von Neumann entropy defined as,
 \begin{equation}
     S_A = -  \Tr( \rho_A \cdot \log(\rho_A)),
     \label{entr}
 \end{equation}
for a subsystem A, yields information about the degree of entanglement between A and B. In this case the entropy is $\log(2)$, which is the largest possible value for a two qubit system. In this thesis, as is often the case, will the entropy of qubits be given in terms of $\log(2)$, so the entropy would be $S_A$ = 1. This value highlights the fact that measuring the state of A will also yield knowledge of the state of B, causing qubit B to collapse to that measured energy level. In other words, the two qubits of an EPR pair are maximally entangled. In this example it holds true that the entanglement entropy of qubit A is equal to the entanglement entropy of qubit B. This property can be generalized to any pure quantum state partitioned into $A$ and $A^c$ and yields $S_{A} = S_{A^c}$. Lastly, the mutual information $I_{AB}$ between subsystems $A$ and $B$ is defined as $I_{AB} = S_A$ + $S_B$ - $S_{AB}$. Because the sum of entropy of the individual systems will always be larger or equal to the entropy of the combination, the mutual information will be larger or equal to zero. \\ \\ 
By exploiting the entanglement between qubits it is possible to  perform computations that classical Turing Machines are not able to do efficiently. Such an example is simulating inherent quantum systems, like the output of a pseudo-random quantum circuit as shown by Google's quantum processor \cite{Google}. Another computational task is the factoring of large numbers into two prime numbers, a task commonly used in the security of online exchanges. For this purpose an efficient quantum algorithm was advanced by P. Shor in 1995 \cite{Shor}.
\subsection{Quantum Error Correction}
\label{try}
 \label{error cor}
The main limitation of quantum processors is their vulnerability to noise. Such noise is usually qubit interacting with its environment. Examples of this include noisy quantum gates or spontaneous emission of the qubit \cite{Got}. In general, quantum errors can be thought of as the qubit transforming to a mixed state entangled with its environment. Today, this noise is still significant in quantum processors. For this reason quantum computing is still classified to be in its noisy intermediate-scale quantum stage or NISQ. One might think to solve this problem with the same error correction framework as for classical computers. Unfortunately, quantum theory prevents this. Specifically, the state of a qubit cannot be cloned. Furthermore, any attempt to measure the state of a qubit will destroy the quantum information contained in that qubit. Lastly, as opposed to discrete errors in regular qubits, quantum errors are continuous  \cite{Nielsen}. As an example, in Formula (\ref{general}) the parameter $\alpha$ might change by an infinitesimal amount. Nevertheless, quantum error correcting codes have been developed. \\ \\
These errors are described using quantum noise operations $\varepsilon (\rho)$ that act on the density matrix $\rho$ of the system. All trace-preserving quantum noise operations $\varepsilon (\rho)$ associated with a closed system, may be approximated with error-operation elements acting on a single qubit, $E_i$, which can then be expressed as a linear combination \cite{Nielsen}, 
\begin{equation}
    E_i = e_{i0} I+  e_{i1} X + e_{i2} Z + e_{i3} X Z,
\end{equation}
where $X$ and $Z$ are the Pauli operators acting on a single qubit. The combination $XZ$ represents the operators acting in right to left order on the qubit. These Pauli operators are,
\begin{equation}
X = \begin{pmatrix}
0 & 1\\
1 & 0 
\end{pmatrix}, \ Y = \begin{pmatrix}
0 & -i\\
i & 0 \end{pmatrix},\ Z = \begin{pmatrix}
1 & 0\\
0 & -1 \end{pmatrix},
\label{paul}
\end{equation}
where the operators $X$ and $Z$ will represent a bit flip and phase shift respectively and $Y = iXZ$. \\
The constants $e_{i0},e_{i1},e_{i2},e_{i3} \in \mathbb{C}$ are sufficient to describe any single operation element $E_i$. These will be able to describe any arbitrary error acting on a single qubit. As an example, \textit{amplitude damping} might occur in a qubit interacting with its environment. This phenomenon can be modeled with a constant $\gamma \in [0,1] $ existing in two operation elements \cite{Nielsen},
\begin{equation}
    E_0 = \begin{pmatrix} 1 & 0 \\
    0 & \sqrt{1-\gamma} \end{pmatrix}, \ E_1 = \begin{pmatrix} 0 & \sqrt{\gamma} \\
    0 & 0 \end{pmatrix}
\end{equation}
acting on a density matrix $\rho$,
\begin{equation}
\varepsilon(\rho) = E_0 \rho E_0^\dag + E_1\rho E_1^\dag.
\label{operation}
\end{equation}
It is crucial to note that two assumptions have already been made. Firstly, it is assumed that the error operation is trace preserving, which corresponds to a closed system. The operation described in Formula (\ref{operation}) is indeed trace-preserving for all density matrices. However, measurements are an example of operations which can be described by a set of single qubit operation elements $E_i$ that are generally not trace-preserving. The second assumption made with this description, is that correlated error operations will not occur on multiple qubits at once. This is justified by the reasoning that if the chance of an error occurring on one qubit is small $\mathcal{O}(\epsilon)$, the chance of a correlated error acting on different qubits at the same time, is even smaller $\mathcal{O}(\epsilon^2)$. \\ The general framework of quantum error correction consists of performing syndrome measurements to detect which error occurred without destroying the quantum information, followed by the corresponding correction operation. The combination of both of these is usually denoted by $R$. A powerful theorem in quantum error correction states that when an error correcting code yields an error correcting procedure $R$ which allows to recover from a trace-preserving noise operation $\varepsilon$ with operation elements $\set{E_i}$, that $R$ will also be able correct the effects of any noise model $\mathcal{F}$ for which its operation elements $\set{F_i}$ are linear combinations of $\set{E_i}$ \cite{Nielsen}.
Consequently, since any operation element is always a linear combination of $X$, $Z$ and $XZ = -iY$, an additional third assumption is justified. The errors are assumed to be only the Pauli matrices X,Z and Y denoted in Equations (\ref{paul}). With these assumptions, one can describe the quantum error correcting codes using the stabilizer formalism.
\subsubsection*{Stabilizer Codes}
 The stabilizer formalism, eloquently formulated by D. Göttesman \cite{Got}, has it roots in group theory. Firstly, analogous to classical linear codes, a quantum code C has three properties $[n,k,d]$, where $n$ will represent the number of encoding or noisy qubits and $k$ the number of error corrected logical qubits. A distance $d$ will be the smallest amount of noisy qubits an error needs to act on, non-trivially, to change one logical state into another. \\
 Suppose now one wishes to use $n$ noisy qubits to compute $k$ fault-tolerant qubits. As previously explained, protecting against the basis of the Pauli $X$, $Y$ and $Z$ matrices and the identity acting on every single qubit, will be enough to also protect from linear combinations of those. This leads to the introduction of the set,
 \begin{equation}
 \mathcal{G} = \set{ \pm I, \pm i I, \pm X, \pm i X,\pm Y, \pm i Y,\pm Z, \pm i Z},
\end{equation}
which will form a group, the Pauli group $\mathcal{G}$, under multiplication. The tensor product of $n$ Pauli groups acting on $n$ qubits is denoted $\mathcal{G}_n$. In this group each element will either commute or anti-commute with each other. A stabilizer group representing an error correcting code is defined as a subgroup $S$ of $\mathcal{G}_n$ that does not contain the negative identity $- (I \otimes I ... \otimes I) $. The stabilizer group defines an action on the $2^n$ dimensional Hilbert space of noisy qubits $\mathcal{H}^{\otimes n}_2$,
\begin{equation}
    S \times \mathcal{H}^{\otimes n}_2 \to \mathcal{H}^{\otimes n}_2, (M,\ket{\psi}) \to M \cdot \ket{\psi}.
\end{equation}
The subspace of the Hilbert space $\mathcal{H}^{\otimes n}_2$ made up of states which stabilizes under $S$, defined as,
\begin{equation}
    P = \set{\ket{\psi} \in \mathcal{H}^{\otimes n}_2 | M \ket{\psi} =\ket{\psi},  \forall M \in S},
\end{equation}
will subsequently span the $2^k$ dimensional Hilbert space of fault-tolerant qubits. These are also called logical qubits \cite{Got}. Lastly, it is possible to write the stabilizer code $S$ more compactly using generators. S will then be the smallest group that contains those generator elements.
\subsubsection*{The Five-Qubit Code}
The five-qubit code is an example of a stabilizer code with properties [5,1,3], meaning it will encode five noisy qubits into one logical qubit. The distance three is the minimum amount of noisy qubits for which an error needs to act on that will transform one logical state into another, which in this case of one logical qubit is flipping the logical qubit.
\begin{table}[h!]
    \centering
    \begin{tabular}{c|c}
        $M_1$ & $X \otimes Z \otimes Z \otimes X \otimes I$ \\
        $M_2$ & $I \otimes X \otimes Z \otimes Z \otimes X$ \\
        $M_3$ & $X \otimes I \otimes X \otimes Z \otimes Z$ \\
        $M_4$ & $Z \otimes X \otimes I \otimes X \otimes Z$ \\
    \hline
    \end{tabular}
    \caption{The generators of the stabilizer group of the five-qubit code \cite{Got}.}
    \label{1}
\end{table}
\newpage \noindent
The generators of the stabilizer group representing this code are given in Table \ref{1}. By construction all generators in this table will leave the code subspace $P$ unchanged. The five-qubit code constructs this code subspace $P$ of the logical qubit as \cite{Got}, \\ \\
\indent $\overline{\ket{0}} = \frac{1}{4} (\ket{00000} + \ket{10010} + \ket{01001} + \ket{10100} + \ket{01010} - \ket{11011} - \ket{00110} - \ket{11000} $ \\ \indent \indent
    $ - \ket{11101} - \ket{00011} - \ket{11110} - \ket{01111}  - \ket{10001} - \ket{01100} - \ket{10111} + \ket{00101} )$,
    \newline\newline
    \indent
     $\overline{\ket{1}} = \frac{1}{4} (\ket{11111} + \ket{01101} + \ket{10110} + \ket{01011} + \ket{10101} - \ket{00100} - \ket{11001}  - \ket{00111}$ \\ \indent \indent $ - \ket{00010} - \ket{11100} - \ket{00001} - \ket{10000} - \ket{01110} - \ket{10011} - \ket{01000} + \ket{11010} )$, \\ \\ 
\noindent
where $\overline{\ket{0}}$ and $\overline{\ket{1}}$ represent the two states of the logical qubit.
The five-qubit code has an additional property that makes it interesting. It is the qubit code with the smallest amount of noisy qubits that is classified as perfect \cite{Got}. This means that all possible error measurements correspond to a correctable error. These error measurements are operators that either commute or anti-commute with all stabilizer elements. For the five-qubit code these are all possible single qubit errors $X_{i}$, $Y_{i}$ and $Z_{i}$. This means that, while the five-qubit code will protect against all single qubit errors, it will not protect against errors acting on multiple qubits at a time. As shown in Section \ref{massage}, the holographic state associated with this code is defined as $\frac{\ket{0 \overline{0}} + \ket{1 \overline{1}}}{\sqrt{2}}$.
\subsection{Tensors}
\label{massage}
Quantum error correcting codes can be represented as a tensor $T$ mapping a collection of $k$ logical qubits to a bigger collection of $n$ noise-susceptible qubits,
\begin{equation}
    T: \mathcal{H}^{\otimes k}_2 \to \mathcal{H}^{\otimes n}_2, \quad \ket{q_1, q_2, ...., q_k} \to \sum^{1}_{p_1,p_2, ...., p_n = 0} T_{p_1,p_2, ...., p_n}^{q_1, q_2, ...., q_k} \ket{p_1,p_2, ...,p_n},
    \label{tensor}
\end{equation}
where $\ket{q_1, q_2, ...., q_k}$ represents a logical state and $\ket{p_1,p_2, ...,p_n}$ a state of noisy qubits. Note that this definition can be generalized to $\nu$ level quantum bits, by changing the number above the sum to $\nu - 1$ \cite{berlin}. As an example, consider the five-qubit code, where the number of logical qubits is one and the number of noisy qubits is five. In this case, Equation (\ref{tensor}) simplifies to,
\begin{equation}
    T: \mathcal{H}_2 \to \mathcal{H}^{\otimes 5}_2, \quad \ket{q} \to \sum^{1}_{a,b,c,d,e = 0} T_{abcde}^{q} \ket{abcde},
\end{equation}
 where values $ T_{abcde}^{q}$ can be derived from the logical qubit states of the five-qubit code, noted in Section \ref{try}. This tensor will be the building block used for the holographic states in this thesis.
 \subsubsection*{Isometric Tensors}
 A linear mapping between any two Hilbert spaces is called an isometry when the inner product between all vectors is conserved. For a tensor acting on finite Hilbert spaces, this is equivalent to stating that $T^{\dag} T$ is the identity on the domain of $T$. Equivalently, it can also be stated as,
 \begin{equation}
     \sum_{a,b,c,d,e} ({T^{q'}_{abcde}})^{\dag} \cdot  T^{q}_{abcde}  = \delta_{qq'}.
 \end{equation} \\
 Isometric tensors have many interesting properties that will be exploited in this thesis. Firstly, in order for a tensor $T : \mathcal{H}^{A} \to \mathcal{H}^{B}$ to be possibly isometric, their dimensions must satisfy dim($\mathcal{H}^{A}$) $\leq$ dim($\mathcal{H}^{B}$). Secondly, isometric tensors allow for easy contractions between itself and its hermitian adjoined.  This is visible in Figure \ref{isometric}, where T and $T^{\dag}$ are executed consecutively. Due to the isometric nature of this tensor, this will be equivalent to the identity operator acting on $A$.
 \begin{figure}[h!]
     \centering
     \includegraphics[scale = 0.75]{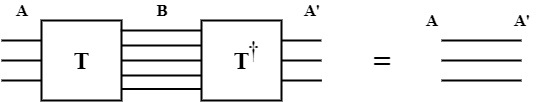}
     \caption{The Penrose graphical notation of an isometric tensor $T$ contracting the five outgoing legs with its hermitian $T^{\dag}$ to create the maximally entangled state between $A$ and $A'$.}
     \label{isometric}
 \end{figure}
While no actual state has been defined in the Hilbert spaces $A$ and $A'$, the holographic state associated to this tensor can be found by reshaping the contracted tensors to a density matrix. For isometric tensors this density matrix, and consequently the subspace A, will always be maximally entangled. However, when actually assigning a state to $A$ the entanglement entropy might change. As an example, projecting all three legs of $A$ to be the ground level $\ket{0}$, will create the state on the Hilbert spaces $A \otimes A'$ as $\ket{000000}$ for which the Von Neumann entropy $S_A$ is zero. 
\\ \\
 Lastly, isometric tensors have a property that is at the core of this thesis. It is the ability to \textit{push} operators through. Consider an operator $O$ acting on the domain of an isometric tensor $T$. This operator is able to be pushed through and transformed to an operator $O'$ acting on the image of the tensor by performing $OT = T T^{\dag} O T $ = $ T O'$, where $O' = T^{\dag} O T$. This is a feature used by quantum error correcting codes to push logical operators to operators acting on the noisy qubits. As an example the logical $\overline{X}$ operator of the five-qubit code can be pushed through the isometric five-qubit tensor and transformed into $X_1 X_2 X_3 X_4 X_5$ \cite{Got}.
\subsubsection*{Perfect Tensors}
Perfect tensors are a subset of isometric tensors for which an additional property holds. Any bipartition of the indices of a tensor $T_{a_1,a_2,..a_n}$, regardless of them being an upper or lower index, into $A$ and $A^c$ with dim($A$) $\leq$ dim($A^c$), will represent a tensor $T$: $A \to A^c$ that is proportional to an isometric tensor. Using the properties of isometric tensors, this allows to state that when a tensor is split up into $A$ smaller or equal than its complement $A^c$, that $TT^{\dag}$ will be proportional the identity.  For example, any bipartition of the indices of the perfect five-qubit tensor, regardless of them representing a logical state or not, into $A$ (incoming legs) and $A^c$ (outgoing legs) where dim($A$) is one, two or three, will have $T T^{\dag}$ be proportional to the identity. Using the fact that $T^{\dag}$ is just the transpose of $T$, this allows for the contraction of subsequent five-qubit code tensors, into a tensor proportional to the identity, as long as there are three or more connected legs. The holographic states that can be associated with these tensors are \textit{absolutely maximally entangled} (AME) states \cite{berlin}. This means that any bipartition of the state into $A$ and $A^c$ with dim($A$) $\leq$ dim($A^c$) will result in a maximally entangled state $A$. As an example, the addition of any leg of the five-qubit tensor to the region $A$ with dim(A) $\leq$ dim($A^c$) will result in the rise of the entanglement entropy $S_A$ with one. Furthermore, due to the property $S_A$ = $S_{A^c}$, this also means that the addition of any leg to region $A$ with dim(A) $\geq$ dim($A^c$) will decrease entanglement entropy $S_A$ by one. Yet again, actually assigning a quantum state to the Hilbert space $A$ might decrease the entanglement entropy. This will be further explored in Section \ref{3}.
\subsubsection*{Tensor Networks}
Using these perfect tensors, it is possible to create networks with certain desirable mathematical properties. As an example, in Figure \ref{pentagon} a network is created from five-qubit tensors represented by blue pentagons. In addition, red dots represent the outgoing logical/bulk qubits and white dots at the boundary the pure states of the noisy/boundary qubits. This is a so called HaPPY code, named after the initials of the authors \cite{Happy}.
\begin{figure}[h!]
    \centering
    \includegraphics[scale=0.75]{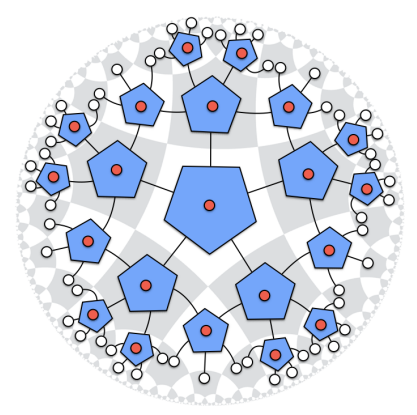}
    \caption{A two layer holographic pentagon code, where blue pentagons represent the five-qubit tensor, red dots the outgoing logical qubits and white dots at the boundary the pure states of the noisy qubits. Figure is taken from \cite{Happy}. }
    \label{pentagon}
\end{figure}
The tensors, more precisely their underlying tessellations, are arranged such that a hyperbolic space is created around a central qubit in the middle. Moreover, the negative curvature of the tessellations of the disk around a central tile together with its maximal symmetry constitutes it being an Anti-de Sitter space or AdS \cite{Happy}. Therefore can these tensor networks be considered as toy models of AdS. The maximal distance between the central tensor and boundary legs, defined as the amount of tensors one needs to pass through to get there, is called the size $n$ of the code. It is also said that the code has $n$ layers. Here the space is two dimensional. However, this formalism can be extended to any dimension \cite{Happy}. \newline \newline
What is often studied with these tensor networks is the ability to reconstruct operators acting on bulk qubits as operators acting  on certain subregion $A$ of boundary qubits. This has a lot of similarities with determining the equivalent noisy qubits operators of logical operators, as done in Section \ref{error cor}. For this purpose is the entanglement entropy, as defined in Formula (\ref{entr}) of the subregion $A$, relevant. For an arbitrary region $A$ a cut $c$ through the tensor network ending on the boundary $\partial A$ between $A$ and $A^c$ may be drawn. The length of this cut is defined as the amount of tensor legs it passes through. It has been found that the entanglement entropy of a boundary region $A$ will satisfy the bound \cite{berlin},
\begin{equation}
    S_A \leq |\gamma_A| \log(\nu),
    \label{entr 2}
\end{equation}
where $|\gamma_A|$ is the length of the smallest cut $\gamma_A$ on the boundary of $A$ and $\log(\nu)$ the natural log of the dimension of the boundary states, which for qubits is $\nu = 2$.
\subsection{AdS/CFT Correspondence}
An attentive reader might wonder why this whole set-up has any relevancy to modern physics. For this the most curious objects in the universe, black holes, are to be examined. In 1974 Stephen Hawking put forth the Bekenstein–Hawking formula for the entropy of a black hole \cite{microsoft},
\begin{equation}
    S_{BH} = \frac{c^3 A}{4 G \hbar},
    \label{bh}
\end{equation}
where $A$ is the surface area of the black hole and $c$, $G$ and $\hbar$ are the speed of light, gravitational constant and Planck's constant divided by $2\pi$, respectively. This essentially states that the entropy of a black hole does not increase with its volume but with its boundary size, the surface area $A$. The similarity between this observation and the bulk/boundary qubits of last chapter, seems to suggest that the quantum states of the black hole are encoded in its boundary states. It was indeed conjectured by Leonard Susskind and Gerardus ’t Hooft that a theory of quantum gravity, not only black holes, would have to be able to reduce a 3+1-dimensional space-time to a 2+1-dimensional description on its boundary \cite{sus, Thooft}. This is called the \textit{holographic principle}. \\ \\
The first concrete realization of the principle was proposed by J. Maldacena in 1997. He conjectured the Anti-de Sitter/Conformal field theory correspondence or AdS/CFT in short, which relates a $d+1$-dimensional gravity theory to a $d$-dimensional quantum field theory on its boundary \cite{maradona}. This theory has been widely discussed, reaching 23.000 citations in 2023. Toy models of this correspondence include the two layer pentagon code, shown in Figure \ref{pentagon}, and other HaPPY codes, such as the single qubit code and the pentagon/hexagon code shown in Figure \ref{temp}. In these models, the bulk/boundary degrees of freedom will correspond to physical/logical degrees of freedom \cite{Happy}. Within this correspondence a different entropy, the Ryu-Takayanagi or RT entropy, determines the entanglement entropy of a connected boundary subsystem $A$ as \cite{nani},
\begin{equation}
    S_{A} = \frac{|\gamma_A|}{4G},
    \label{entr3}
\end{equation}
where $|\gamma_A|$ is the minimal area in Planck units ($l_p$) of a $d-1$-dimensional surface $\gamma_A$ with its boundary on $A$ and $G$ the gravitational constant. Similarly, will the entropy of subregions $A$ of tensor networks in this AdS space be calculated analogously to this RT formula, where $|\gamma_A|$ becomes the minimal cut through this tensor network ending on $A$. The RT Formula is not only a generalization of the black hole entropy, Formula (\ref{bh}), but it is also in agreement with this bound in entropy of Formula (\ref{entr 2}). The AdS/CFT correspondence has already experienced success in modeling complex gravitation systems such as black holes and wormholes \cite{bbh, google}.
\begin{figure}[h!]
    \centering
    \includegraphics[scale=0.66]{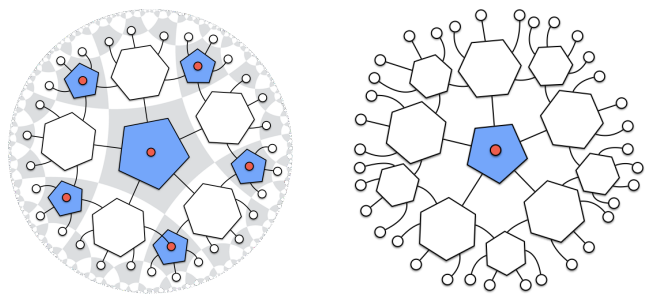}
    \caption{Pentagon/hexagon code (left) and single qubit code (right). Figure from \cite{Happy}.}
    \label{temp}
\end{figure}
\section{Wedges and Operator Reconstructability}
\label{sectionwed}
The purpose for which the introduced tools are utilized is operator reconstructability in AdS/CFT. In this theory a $d+1$-dimensional gravity bulk is modeled with $d$-dimensional quantum field theory on its boundary. Consequently, operations in this gravitational bulk correspond to operators on its boundary. This is usually done by pushing the operators outwards to the boundary, analogously to the method introduced in Section \ref{massage}. However, when working with a fixed subregion $A$ of the boundary, bulk operations might not allow to be modeled with only $A$. \\
\\ For this purpose, a number of wedges will be introduced which correspond to cuts within the bulk of tensor networks that encapsulates bulk qubits. These include the causal wedge, the greedy entanglement wedge, the minimal entanglement wedge and the geometric entanglement wedge. In the next section the differences between these wedges will be examined. Additionally, it is investigated how changes in the tensor network such as different types of networks, non uniform boundaries and stacking of the networks, impact the shape of the wedges.\\ \\
The framework to study these changes will be Monte-Carlo simulations of the boundary erasure and central qubit inclusion. Specifically, for every iteration of the simulation a new erasure probability $p$ will be uniformly chosen between zero and one. With this probability a boundary subregion $A$ will be established by performing a Bernoulli trial for every boundary leg with the probability $p$ where success corresponds to the erasure of the boundary qubit and failure to the inclusion of the boundary qubit into $A$. Thereafter, the respective wedge will be constructed on the boundary of $A$. If the wedge contains the central bulk qubit, defined as the bulk qubit of the central tensor in the tensor network, the simulation will be deemed a success. Finally, the success ratio of the simulation will be averaged over 30 even intervals of $p$, corresponding to bins, and plotted in function of the middle of those bins. This method was first introduced by \cite{Happy}.  The method allows for the detection of threshold probabilities $p_c$ defined as the maximum erasure probability where the chance of the central qubit not included in a defined wedge, becomes exponentially small with increasing $n$ for all $p < p_c$. The threshold probability will be the erasure probability for which the plot will converge to a step function for $n$ going to infinity. 
\begin{figure}[h!]
    \centering
    \includegraphics[scale = 0.35]{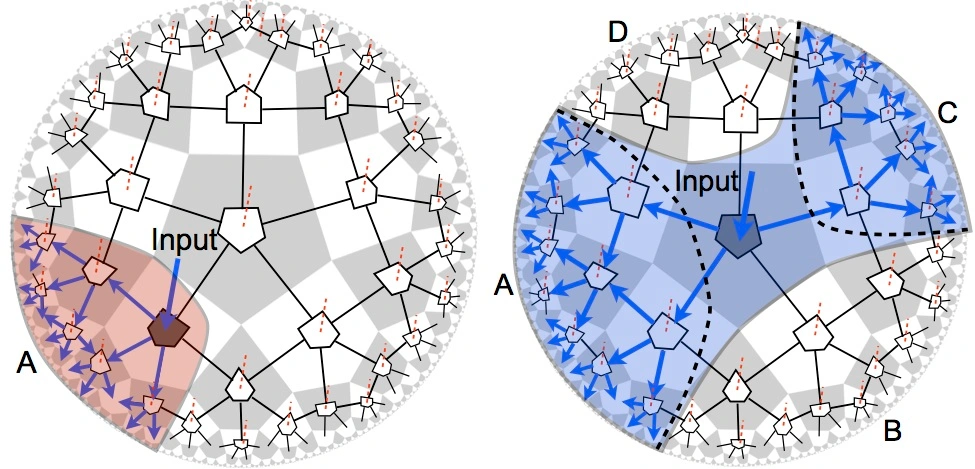}
    \caption{Example of the causal wedge (left) and the entanglement wedge (right) indicating the operator reconstructability in the holographic pentagon code of a bulk qubit, labeled input, and the boundary of respectively $A$ and $A \cup C$. Figure is taken from \cite{Yoshi}.}
    \label{wedges}
\end{figure}
\subsection{Causal Wedge}
The causal wedge received its name from its bulk gravitational interpretation. It is defined in \cite{caus} for a boundary subregion $A$ as the intersection of the bulk past and future of the domain of dependence of $A$. This domain of dependence is the region of bulk points where their physics is fully determined by the initial conditions of $A$. However, what is often investigated in literature, is the reduced causal wedge where the wedge is projected to the same time slice as $A$. Hence, the initial wedge shape in the space-time diagram is reduced to a space slice visible in Figure \ref{wedges} (left). Additionally, in the setting of tensor networks, the construction of this causal wedge will not rely on causal connections, but will rather rely on extremizing a $d-1$-dimensional surface $\gamma_A$ in accordance with the RT entropy depicted in Formula (\ref{entr3}). This construction was described by Pastawski et al. \cite{Happy} and relies on performing a so called greedy algorithm on the boundary. \\ \\ This algorithm consists of first performing a trivial cut $c$ along the boundary $A$ and then extending that cut each iteration. In an iteration all tensors outside the cut that have at least half of their legs connected to tensors within the cut will force the cut to move and include the tensor to the region. The algorithm will end in the resulting greedy wedge $\gamma^*_A$ when no new tensors are added in one iteration. Because the tensors added to cut $c$ had at least half their legs intersecting with $c$, the new cut will always have a smaller or equal size. Consequently, $|\gamma^*_A|$ will be a local minimum of the cut sizes. As stated in \cite{Happy}, that local minimum will be an absolute minimum, resulting in $\gamma^*_A = \gamma_A$, when $A$ is a connected region on the boundary. However, when the region is disconnected, the causal greedy algorithm may not succeed in finding the true minimal geodesic or geometric wedge $\gamma_A$. This is visible in Figure \ref{wedges} (right), where the greedy algorithm applied to the disconnected region $A$ and $C$ individually, highlighted by the black dotted line, does not find the geometric wedge $\gamma_A$ which encapsulates the blue region. The Monte-Carlo simulation for the holographic pentagon code of the central qubit inclusion when performing the greedy algorithm for each disconnected region resulting from the simulated erasure is visible in Figure \ref{sim-1}.
\begin{figure}[h!]
    \centering
    \includegraphics[scale = 0.69]{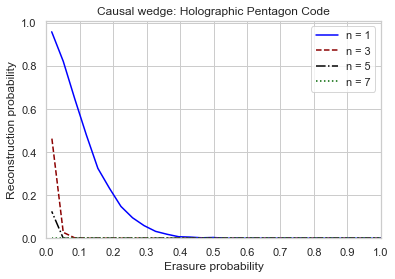}
    \caption{A Monte-Carlo simulation of the central tensor being incorporated by the greedy causal algorithm applied to the holographic pentagon code with $n$ layers. The size of the Monte-Carlo simulation is $10^5$ for the one and three layer code and $250$ for the five and seven layer code.}
    \label{sim-1}
\end{figure}
\\ \noindent
In Figure \ref{sim-1} it is visible that there is no convergence of the inclusion curve to a step function, but rather it converges to the zero function corresponding to an erasure probability $p_c$ of zero. In other words, the erasure of an infinitesimal small percentage of boundary qubits will, in the large $n$ limit, almost always result in the omission of the central qubit in the causal wedge. It has been found in \cite{Happy} that for the pentagon code the amount of boundary qubits rises as $\mathcal{O}\left(\left(\frac{3+\sqrt5}{2}\right)^n\right)$ with the amount of layers $n$. Consequently, Figure \ref{sim-1} indicates that the number of boundary qubits that are allowed to be erased while still maintaining central qubit inclusion rises significantly slower than that. \\ \\
This is in agreement with a clever example in \cite{Happy}, where it is stated that the erasure of just four boundary qubits at certain places in a large layer pentagon code will result in the omission of the central qubit. The behaviour of the inclusion graph is a result of the disconnected regions, created due to the erasure of just a couple of qubits, not being able to individually attain the true minimal bulk geodesic $\gamma_A$. Moreover, simulations with the single qubit code and the pentagon/hexagon code showed the same convergence behaviour.
\subsection{Entanglement Wedge}
To be compatible with disconnected regions of the boundary, the greedy entanglement wedge was introduced by Pastawski et al. \cite{Happy}. Similar to the causal wedge, is the greedy entanglement wedge constructed by applying the greedy algorithm to the boundary. However, rather than applying the algorithm to each disconnected boundary region individually, the algorithm will now be applied to all regions simultaneously. As a result, disconnected regions are now able to contribute together to advancing the cut. The greedy entanglement wedge is visible in Figure \ref{wedges} (right). Here, by applying the greedy entanglement algorithm to boundary regions $A$ and $C$, the greedy entanglement wedge $\gamma_A^*$ that encapsulates the blue region is created. For this example this is also the geometric wedge $\gamma_A$. On this figure, the distinction between the causal- and entanglement wedge is clearly visible. The entanglement wedge includes the central tensors and the individual causal wedges, noted by the black dotted line, fail to include it. \\ \\
This difference is also visible in the Monte-Carlo simulation for the boundary erasure/central qubit inclusion in the holographic pentagon network visible in Figure \ref{sim0}. In this figure, it is visible that the graphs are descending significantly later than the causal wedge. However, by looking at probabilities where the graphs start to drop, the inclusion graph seems to still be converging to the zero function. Therefore, the erasure of any fraction of the boundary qubits will, in the big $n$ limit, almost always result in the omission of the central qubit. In other words no threshold probability is detected. This conclusion and simulation is in agreement with the same Monte-Carlo simulation performed by Pastawski et al. \cite{Happy}. \\ \\The question may arise if the individual positions of the erased boundary qubits have any significance in the inclusion of the central qubit. Figure \ref{sim0} indicates this as well. For small layer networks the function gradually drops to zero, indicating that for the same erasure probability there are some configurations of $A$ which include the central qubit and some which do not. For the pentagon code this can be illustrated by the example in Figure \ref{exmp1}. For this one layer pentagon network, a different boundary $A$ of the same size has been created. On the left was for three tensors only one qubit erased. Due to three qubits still remaining in the boundary $A$, the greedy entanglement wedge will still advance past all these tensors, eventually including the central qubit. This is due to the fact that $A$ has a very top heavy shape. However, on the right an identical network with the same boundary size was created where the central qubit is not included in the entanglement wedge. Although these very asymmetric configurations of $A$ are still possible in networks with more layers, they become exponentially unlikely as $n$ increases. Therefore, in Figure \ref{sim0} the inclusion graphs drop more rapidly for higher layer networks.
\begin{figure}[h!]
    \centering
    \includegraphics[scale = 0.69]{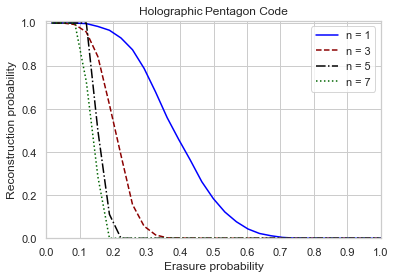}
    \caption{A Monte-Carlo simulation of the central tensor being incorporated by the greedy entanglement algorithm applied to the holographic pentagon code with $n$ layers. The size of the Monte-Carlo simulation is $10^5$ for the one and three layer code and $250$ for the five and seven layer code.}
    \label{sim0}
\end{figure}
\begin{figure}[h!]
\minipage{0.50\textwidth}
  \includegraphics[width=\linewidth]{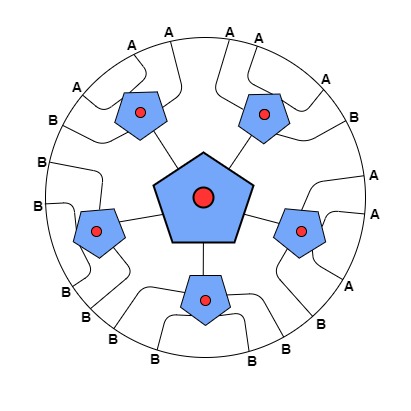}
\endminipage\hfill
\minipage{0.50\textwidth}%
  \includegraphics[width=\linewidth]{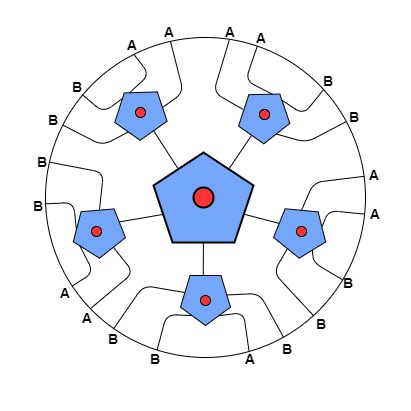}
\endminipage
\caption{An example of the shape of the boundary region $A$ having an impact on the inclusion of the central qubit in the greedy entanglement wedge. On the left the greedy entanglement wedge will advance to include the central qubit, where on the right it will not.}
\label{exmp1}
\end{figure}
\\ \\
\subsubsection*{Type of tensor networks}
It can be investigated how changes in the type of tensor network impact a possible threshold erasure probability $p_c$. For this purpose, the tensor networks visible in Figure \ref{temp} were used. While these tensor networks are still HaPPY codes, satisfying the requirements for the AdS/CFT correspondence, they differ from the holographic pentagon by the use of hexagons. These hexagons still are five-qubit tensors. However, the logical qubit leg, unlike pentagons, do not represent bulk outgoing qubit but rather legs that remain in the AdS-surface. As a result, the underlying structure of the tensor network changes. These changes can be represented by the fraction $f = \frac{N_{bulk}}{N_{boundary}}$, where $N_{bulk}$ and $N_{boundary}$ are the total amount of bulk and boundary qubits respectively. The value of this parameter was found by \cite{Happy} for the different tensor networks. While the fraction may differ based on which layer the boundary was cut off, in general it will hold $f_{pent} < f_{pent/hex} < f_{single qubit}$. In Figure \ref{sim1}, it can be found how this fraction impacts the Monte-Carlo simulations of central qubit inclusion and boundary erasure.
\begin{figure}[!htb]
\minipage{0.50\textwidth}
  \includegraphics[width=\linewidth, height=5cm]{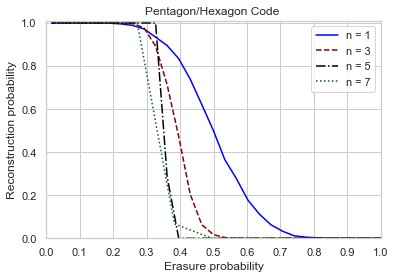}
\endminipage\hfill
\minipage{0.50\textwidth}%
  \includegraphics[width=\linewidth,height=5cm]{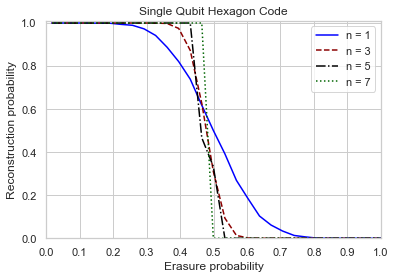}
\endminipage
\caption{A Monte-Carlo simulation of the central tensor being incorporated by the greedy entanglement algorithm applied to the pentagon/hexagon code (left) and the single qubit code (right) with $n$ layers. The size of the Monte-Carlo simulation is $10^5$ for the one and three layer code and $250$ for the five and seven layer code.}
\label{sim1}
\end{figure}
\\
The Monte-Carlo simulations indicate that, in contrast to the causal wedge, the inclusion graphs do change with type of tensor network. Recall that for the holographic pentagon code using the greedy entanglement there was no threshold probability detected by the Monte-Carlo simulation in Figure \ref{sim0}. In Figure \ref{sim1}, it is visible that the single qubit network does have an erasure thresholds around the value $p_c \approx 0.50$. This means that about half of boundary qubits may be erased before the central qubit will be excluded from the greedy entanglement wedge. While for the pentagon/hexagon network it is unsure if a threshold probability exists, it still shows better inclusion than the pentagon network. The results indicate that decreasing the fraction $\frac{N_{bulk}}{N_{boundary}}$ by changing tensor networks will improve the threshold probability $p_c$ until a maximum value of 0.5. \\ \\
One hypothesis to possibly justify this behaviour is that a lower fraction allows the mutual information between the central qubit and boundary region to be more spread out across the boundary. Therefore, some minimal value of mutual information could be reached more quickly for a lower fraction. Consequently, the operators acting on the central qubit could allow for a faster reconstructability. This idea will be further investigated in Section \ref{3}.
\subsubsection*{Shape of the boundary}
As suggested by \cite{Happy}, another way to decrease the fraction $\frac{N_{bulk}}{N_{boundary}}$ is by performing a non-uniform cut off on the boundary. This is executed by allowing boundary qubits to have different distances to the central tensor. In Figure \ref{temp}, such a non-uniform cut off is visible of the pentagon/hexagon- and single qubit code. For these networks, the cut off is executed by removing the tensors which are not connected to two or more tensors of the previous layer. The Monte-Carlo simulations for these networks are visible in Figure \ref{sim2}. These results are consistent with the simulations performed by \cite{Happy}. By comparing these inclusion graphs with the equivalent uniform boundary graphs depicted in Figure \ref{sim1}, it is visible that a change in boundary shape has a small effect on the inclusion graph. A significant change in shape occurs for the three layer tensor networks, where the graph seems to be more spread out with a non uniform boundary. Because the cut off for the one layer networks remains the same, the three layer networks are the smallest simulated tensor networks that are effected by the change in boundary. A three layer network with a non-uniform cut off will have less boundary legs and consequently will be more susceptible to asymmetric boundary division, resulting in the stretching of the inclusion graph.
\begin{figure}[h!]
\minipage{0.50\textwidth}
  \includegraphics[width=\linewidth,height=5cm]{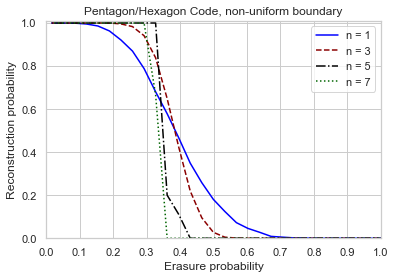}
\endminipage\hfill
\minipage{0.50\textwidth}
  \includegraphics[width=\linewidth,height=5cm]{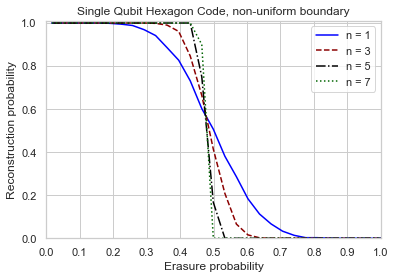}
\endminipage
\caption{A Monte-Carlo simulation of the central tensor being incorporated by the greedy entanglement algorithm applied to the pentagon/hexagon code (left) and the single qubit code (right) with n layers with a non-uniform boundary. The size of the Monte-Carlo simulation is $10^5$ for the one and three layer code and $250$ for the five and seven layer code.}
\label{sim2}
\end{figure}
The inclusion graph of a single qubit network has the same threshold probability, graphically determined as $p_c \approx 0.50$. For the pentagon/hexagon code the plausibility of a threshold probability seems to have increased with respect to the uniform boundary inclusion graph. This is in agreement with \cite{Happy}, where it was proposed that through a bounding procedure it was possible to find a lower bound on the non zero threshold probability of the non-uniform pentagon/hexagon code. For the non-uniform single qubit network, such a lower bound was indeed found by \cite{Happy}. They have proven analytically that for any erasure probability $p$ with $p < \frac{1}{12}$, the probability to have guaranteed operator reconstructability is higher than $1 - \frac{1}{12} \left( 12p \right)^{\lambda^n}$, with $\lambda = \frac{1 + \sqrt{5}}{2}$. This means that the chance of exclusion of the central tensor becomes exponentially small for increasing layer depth $n$ for any $p < \frac{1}{12}$. For this estimation a hierarchical recovery model was used based on the propagation models of erasures on the boundary. This model is weaker than the greedy entanglement wedge since it does not allow for lower layer tensors to contribute to higher layer tensor inclusions. This means that the threshold probability $p_c$ for inclusion in the greedy entanglement wedge is at least larger or equal than $\frac{1}{12}$. However, this model fails with the pentagon code since the recursion relationship between errors of different layers does not allow for a scaling dimension $\lambda$ different than one. This basically means that the error rate does not decrease or increase towards the central qubit for this hierarchical model \cite{Happy}. Therefore, no lower bound on the threshold probability was able to be established. An attempt was made to construct this same hierarchical model for the pentagon/hexagon code, but no threshold probability was found.
\subsubsection*{Stacking of tensor networks}
One last possible method to possibly increase the fraction $\frac{N_{bulk}}{N_{boundary}}$ is to stack multiple tensor networks. This method, first introduced by \cite{stack}, consists of stacking multiple tensor networks on top of each other and routing the corresponding bulk legs to a single bulk leg using random all-to-all (AA) scrambling circuits, illustrated in Figure \ref{stack}. These circuits will be simulated by constructing tensors where all the coefficients are randomly chosen from a standard normal distribution. Through this stacking method, the fraction can be made arbitrarily low by increasing stack size. There are two possible ways of defining a boundary region for stacked tensor networks. The first method is to create a boundary region $A$ for a single stack and copying that boundary region to all stacks. This is the method introduced by \cite{stack} and visible in Figure \ref{stack}. However, one can also define new regions $A$ with the same erasure probability $p$ on each stack of the network. For the first method the greedy entanglement wedge will be identical for all stacks and the resulting inclusion graph will remain the same, whereas with the second method the entanglement wedge might increase to include some of the random all-to-all tensor resulting in a net positive in central qubit inclusions. This is visible in the Monte-Carlo simulations of Figure \ref{sim3}.
\begin{figure}[h!]
    \centering
    \includegraphics[scale = 0.75]{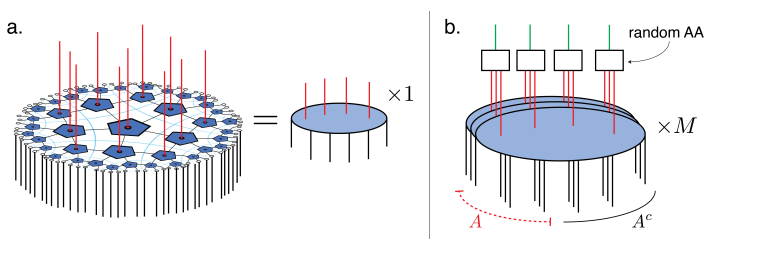}
    \caption{Stacking the pentagon code. On the left a single holographic pentagon code and on the right $M$ stacked pentagon codes connected to random all-to-all scrambling circuits. Bulk legs are highlighted in red, and the boundary is noted with $A$. Figure is taken from \cite{stack}.}
    \label{stack}
\end{figure}
\begin{figure}[h!]
    \centering
    \includegraphics[scale = 0.65]{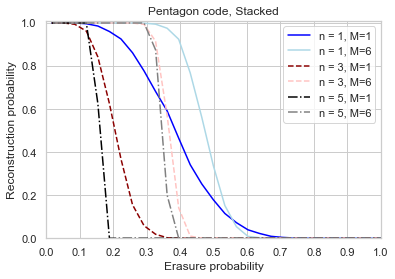}
    \caption{A Monte-Carlo simulation of the central tensor being incorporated by the greedy entanglement algorithm applied to $M$ stacked holographic pentagon codes with n layers, where the boundary subregion is redefined in each stack. The size of the Monte-Carlo simulation is $10^5$ for the one and three layer code and $250$ for the five layer code.}
    \label{sim3}
\end{figure}
 \\
In Figure \ref{sim3} Monte-Carlo simulations performed for boundary erasure/central qubit inclusion with the greedy entanglement wedge of six stacked holographic pentagon codes are visible as the light colors. The inclusion graph has moved from its one stack equivalent to the left by about 0.3 erasure probability. As mentioned this happens due to boundary region $A$ being different for all stacks. This results in some of the random all-to-all (AA) tensors being included in the entanglement wedge. This in turn can allow to entanglement wedge to be pushed further on other stacks. Note as well that only half the legs of the central AA tensor have to be included for the central qubit to be included and the simulation to be deemed successful. However, still no threshold probability is detected for this pentagon code. Lastly, the distribution of the AA tensors has no impact on the success rate. As an example, identical simulations with the coefficients uniformly drawn between zero and one, yielded the same inclusion graph.
\newpage
\noindent
While through this method of stacking, the fraction $\frac{N_{bulk}}{N_{boundary}}$ can be made arbitrarily low, it is visible in the shape of the inclusion graph that there still does not seem to be convergence to a threshold probability. This shows that central qubit recovery is not only affected by $\frac{N_{bulk}}{N_{boundary}}$ but also by the inherent structure of the tensor network. Changing the fraction may move the inclusion graphs of the same networks, but the inherent geometry of the tensor network still determines if there is convergence or not. This is also visible when attempting to do the same lower bound estimation for this stacked pentagon code. By stacking the tensor networks, one is essentially just adding a new ground layer to the tensor network. However, the scaling dimension of the propagation of errors of an hierarchical model of this stacked pentagon tensor network will still remain one, regardless of the fact that at the end a new layer has been added.
\subsection{Minimum Entanglement Wedge}
\label{sim}
The final wedge that is investigated in this thesis is the minimum entanglement wedge. This newly introduced notion of entanglement wedge loses its ensured operator reconstructability but will be a better approximate for the entanglement entropy. Since the greedy algorithm applied to $A$ finds a local minimum in wedge size $|\gamma^*_A|$, the greedy algorithm performed on the compliment $A^c$ will do the same for $\gamma^*_{A^c}$. However, this geodesic $\gamma^*_{A^c}$ will also have its boundary on A. Consequently, the geodesic $\gamma^*_{A^c}$ is an additional candidate for the true geometric entanglement wedge $\gamma_A$. From this idea the minimum entanglement wedge can be created as the minimum in geodesic size,
\begin{equation}
    \gamma_A^{min}  = \min\{{ \gamma_A^* , \gamma_{A^c}^* }\}
\end{equation}
where $\gamma_A^*$ and $\gamma_{A^c}^*$ are the greedy algorithm performed on $A$ and $A^c$ respectively. \\
For any boundary subregion $A$, the length of the cut of this minimum entanglement wedge $\gamma_A^{min}$ will be smaller or equal to the cut of the greedy entanglement wedge, and thus be a better candidate to the minimum cut of the true geometric wedge. This allows for a better approximation of the RT entropy in Formula (\ref{entr3}). However, as further explained in Section \ref{reconstr}, the wedge does not have the guaranteed reconstructability on $A$ for operators acting on bulk qubits included in the wedge. 
\subsubsection*{Classical RT Entropy} 
Before performing this algorithm, a clear idea of wedge size has to be defined. In this thesis the wedge size will be defined as the number of tensor legs the geodesic representing that wedge crosses. This wedge size will be used to calculate the entropy of subregions analogously to the classic RT formula depicted in Formula (\ref{entr3}). However, for error correcting tensor networks the holographic Von Neumann entropy for a boundary $A$ can be approximated with \cite{corr1, corr2},
\begin{equation}
S(\rho_A) = \Tr(\rho \mathcal{L}_A) + S_{bulk}(\rho_{\epsilon_A}),
\label{error corected RT}
\end{equation}
where $\mathcal{L}_A$ is $\frac{Area(\gamma_A)}{4G}$ with the Newtonian constant $G$ and $\rho_{\epsilon_A}$ is the density matrix of all bulk qubits included in the entanglement wedge. \\ \\
As is visible, Formula (\ref{entr3}) and Formula (\ref{error corected RT}) will no longer agree. Consequently their tensor network equivalent will not either. However, numerically it has been found that, similar to the example of Section \ref{massage}, projecting the bulk qubit of the five qubit tensor to $\ket{0}$ or $\ket{1}$ changes the entropy. More precisely, when tracing out any three qubits of the tensor, the entanglement entropy will not increase to three anymore, but rather remain at two. Because all other interactions remain the same, this can be translated to removing all contributions of bulk qubits to the Von Neumann entropy. In other words the second term of Formula (\ref{error corected RT}) vanishes and all that remains is the wedge size, which corresponds to using the classic RT Formula (\ref{entr3}). The entropy of these systems will be further investigated in Section \ref{3}. Consequently, for the algorithm of the minimum wedge, the size of the wedge will be considered as the true size, the amount of tensor legs the geodesic crosses. This means that a geodesic which has its boundary on $A$ will have the same size when looking from the compliment boundary $A^c$.
\begin{figure}[h!]
    \centering
    \includegraphics[scale = 0.64]{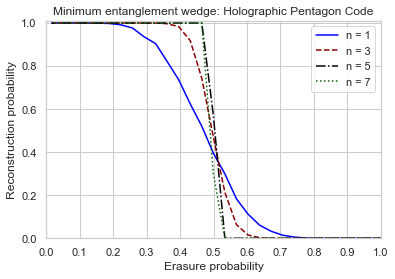}
    \caption{A Monte-Carlo simulation of the central tensor being incorporated by the minimum entanglement wedge with classical size applied to the holographic pentagon code with $n$ layers. The size of the Monte-Carlo simulation is $10^5$ for the one and three layer code and $250$ for the five and seven layer code.}
    \label{sim3.5}
\end{figure}
\noindent
The Monte-Carlo simulation of central qubit inclusion/boundary erasure using the minimum entanglement wedge without bulk leg contributions is visible in Figure \ref{sim3.5} for the holographic pentagon code. Here, the inclusion graph approaches a true step function around the value $p = 0.50$. Consequently, the threshold probability $p_c$ is 0.5. This was defined as the maximum erasure probability where the chance of the central qubit not included in a defined wedge becomes exponentially small with increasing $n$ for all $p < p_c$. This is a property of not only the minimum entanglement wedge acting on the pentagon code, but for the pentagon/hexagon- and single qubit code as well, visible in Figure \ref{sim4}. In an effort to explain this behaviour Theorem (\ref{bigbrain}) was introduced. 
\begin{theorem}
\label{bigbrain}
The threshold probability $p_c$ for any wedge $ \gamma_A $ is 0.50 regardless of the structure of a tensor network as long as the probability of $\gamma_A \neq \gamma_{A^c}$ becomes exponentially small for an increasing amount of layers.
\end{theorem}
\begin{proof}
To proof by contradiction, assume $p_c > 0.5 $. Consequently, there exists an erasure probability $p_c > p > 0.5$, where the probability of exclusion of the central qubit included in the wedge, becomes exponentially small with increasing number of layers $n$. For all regions $A$ originating from this probability $p$ will, due to the assumption, $ \gamma_A $ almost always be equal to $\gamma_{A^c}$. As a result a bulk region enclosed by $\gamma_A$ and $A^c$, originating from an erasure probability $1-p$ has been found for which the probability of exclusion of the central qubit between $A^c$ and $\gamma_{A^c}$ not become exponentially small with $n$. Consequently, because $A$ was chosen arbitrarily, there exists an erasure probability $1-p < 0.5 < p_c$ for which the definition of $p_c$ does not hold.  \\ \\
Conversely, assume that $p_c < 0.5$. For almost every region $A$ created by a erasure probability $p_c < p < 0.5$ that includes the central tensor, there will exist a region enclosed by $\gamma_A$ and $A^c$ which will not include the central qubit. Consequently, an erasure probability $1-p < 0.5 < p_c$ was found where the chance of exclusion of the central qubit does not become exponentially small with increasing $n$, in contradiction with the definition of the threshold probability.
\end{proof}
\begin{figure}[h!]
\minipage{0.50\textwidth}
  \includegraphics[width=\linewidth,height=5cm]{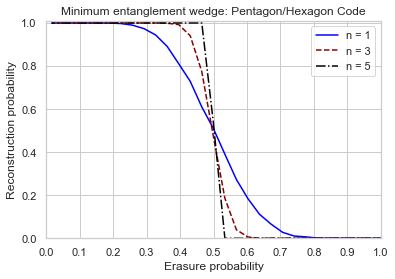}
\endminipage\hfill
\minipage{0.50\textwidth}%
  \includegraphics[width=\linewidth,height=5cm]{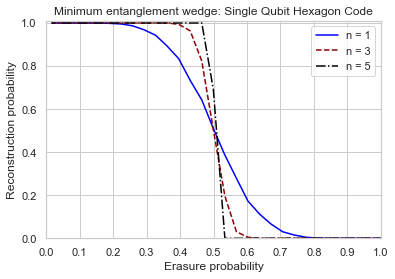}
\endminipage
\caption{A Monte-Carlo simulation of the central tensor being incorporated by the minimum entanglement wedge with classical size of the pentagon/hexagon code (left) and the single qubit code (right) with n layers. The size of the Monte-Carlo simulation is $10^5$ for the one and three layer code and $250$ for the five layer code.}
\label{sim4}
\end{figure}
\noindent This theorem applied to the minimum entanglement wedge states that if the probability $\gamma_A^{min}\neq \gamma_{A^c}^{min}$  becomes exponentially small for an increasing amount of layers, the threshold probability $p_c$ will be 0.50. The minimum entanglement wedge is not the same geodesic for a region and its compliment when the sizes of the greedy entanglement wedges are equal $|\gamma_A^*| = |\gamma_{A^c}^*|$, but the wedges itself are not $\gamma_A^* \neq \gamma_{A^c}^*$. In this case, in accordance to the geometric wedge, the minimum entanglement wedge of $A$ is $\gamma_{A^c}^*$ and of $A^c$ is $\gamma_A^*$ . While a definite proof of $\gamma_A^{min}\neq \gamma_{A^c}^{min}$ becoming exponentially small for increasing $n$ has not been found, an experimental case study, visible in Figure \ref{case}, suggests that it might be true. In this study, similar to the Monte-Carlo simulations, boundary regions $A$ were created corresponding to different erasure probabilities. For these regions the greedy entanglement algorithm was applied to both $A$ and $A^c$. Subsequently different greedy entanglement wedges with the same size, corresponding to different minimal entanglement wedges, were detected. In Figure \ref{case}, this estimated case probability was plotted in 30 bins of erasure probabilities. It is visible that for an arbitrary erasure probability $p$, the probability that $\gamma_A^{min}$ is not equal to $\gamma_{A^c}^{min}$ becomes really small for increasing layer depth $n$.
\begin{figure}[h!]
    \centering
    \includegraphics[scale=0.58]{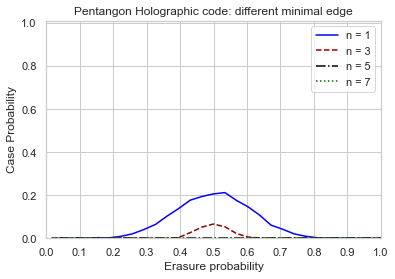}
    \caption{The estimated probabilities of the minimum entanglement wedge applied to the holographic pentagon code being different for a Monte-Carlo simulated boundary subregion $A$ and its complement $A^c$. The size of the Monte-Carlo simulation is $10^5$ for the one and three layer code and $250$ for the five and seven layer code.}
    \label{case}
\end{figure}
\subsubsection*{Error Corrected Entropy}
It is also possible to investigate the minimum entanglement wedge for the holographic state. By definition and numerical study of the five-qubit tensor, this is equivalent to utilizing the holographic state $\frac{\ket{0 \overline{0}} + \ket{1 \overline{1}}}{\sqrt{2}}$ in all tensors. With these states the entanglement entropy will rise with one whenever a bulk qubit is included in the wedge. Consequently, the entropy of the boundary subregion will also rise according to Formula (\ref{error corected RT}). For the algorithm acting on the subregion $A$ it means that for $A$, one will be added to the size of entanglement wedge $|\gamma^{*}_{A}|$ whenever a bulk qubit is included and for $A^c$ one will be added to the size of $|\gamma^{*}_{A^c}|$ for all bulk qubits encapsulated between $A$ and $\gamma^{*}_{A^c}$. As a result, $|\gamma^{min}_{A}|$ and $|\gamma^{min}_{A^c}|$ is no longer guaranteed to be equal.
\begin{figure}[h!]
    \centering
    \includegraphics[scale = 0.64]{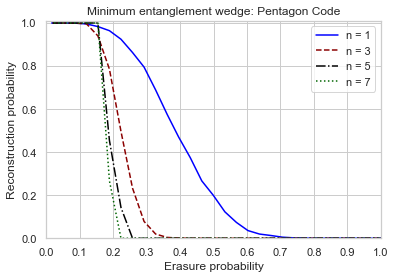}
    \caption{A Monte-Carlo simulation of the central tensor being incorporated by the minimum entanglement wedge with corrected size applied to the holographic pentagon code with $n$ layers. The size of the Monte-Carlo simulation is $10^5$ for the one and three layer code and $250$ for the five and seven layer code.}
    \label{min1}
\end{figure}
\\
In Figure \ref{min1} the central qubit inclusion for simulated boundary erasure with different erasure probabilities is visible for the minimum entanglement wedge that adds one to the wedge size for every bulk qubit included. Comparing this graph with the Monte-Carlo simulation of the entanglement wedge, visible in Figure \ref{sim0}, shows that the inclusion graph in the minimum entanglement wedge has shifted a bit to the left of that of the greedy entanglement wedge. This was expected since the minimum entanglement wedge will either encapsulate more or an equal amount of tensors than the greedy entanglement wedge. However, this shift in probability is only about 0.05. This means that the minimum entanglement wedge will include the central bulk qubit for a fraction of the boundary states that is only slightly higher than that the greedy entanglement wedge. The inclusion graph still does not seem to converge to a step function. In Figure \ref{min2}, the Monte-Carlo simulations for the single qubit- and pentagon/hexagon code are also visible. In this figure it is still visible that the pentagon/hexagon network does seems to converge to a step function with an erasure probability of about $p_c \approx 0.35$. The threshold probability of the single qubit code seems to advance from its entanglement wedge equivalent $p_c \approx 0.50$ to about $p_c \approx 0.60$. This indicates that the minimum entanglement wedge is moving past the greedy entanglement wedge. Lastly, there seem to be significant differences when comparing the inclusion graphs of the classical RT formula and those of the error corrected entropy. This happens since a geodesic which has its boundary on $A$ will no longer have the same size when looking from the compliment boundary $A^c$. Consequently, the proof of Theorem (\ref{bigbrain}) is no longer applicable.
\begin{figure}[h!]
\minipage{0.50\textwidth}
  \includegraphics[width=\linewidth,height=5cm]{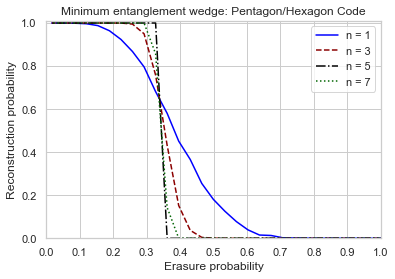}
\endminipage\hfill
\minipage{0.50\textwidth}
  \includegraphics[width=\linewidth,height=5cm]{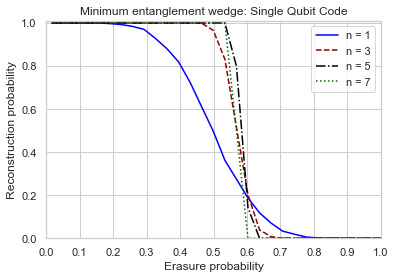}
\endminipage
\caption{A Monte-Carlo simulation of the central tensor being incorporated by the minimum entanglement algorithm with corrected size applied to the pentagon/hexagon code (left) and the single qubit code (right) with n layers. The size of the Monte-Carlo simulation is $10^5$ for the one and three layer code and $250$ for the five and seven layer code.}
\label{min2}
\end{figure}
\subsubsection*{Geometric wedge}
An example of the minimal entanglement wedge reaching the true geometric wedge is visible in Figure \ref{exmp2}. For this single qubit network with a non-uniform boundary the minimum entanglement wedge $\gamma_A^{min}$ reaches the true geometric entanglement wedge $\gamma_A$ as the blue geodesic. This blue geodesic is the greedy algorithm applied to the complement boundary $B$ and has the size $|\gamma_B^*| = 9$. The greedy entanglement wedge of $A$ highlighted in red has a size $|\gamma_A^*| = 12$. Consequently, the greedy entanglement algorithm falls short in finding the true geometric wedge. These wedge sizes have been calculated according to the classic RT formula, as the amount of legs the geodesic crosses. However, the example still holds for any state of the bulk qubit. For this example, an operator is reconstructed beyond the greedy entanglement algorithm in Section \ref{reconstr}.
\begin{figure}[h!]
    \centering
    \includegraphics[scale = 0.52]{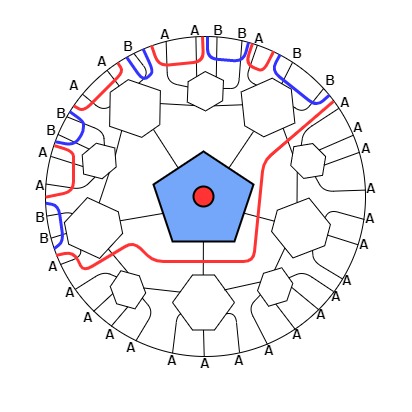}
    \caption{An example of a non-uniform boundary two layer single qubit network where the minimum entanglement wedge $\gamma_A^{min}$ reaches the true geometric entanglement wedge $\gamma_A$ as the blue geodesic and the greedy entanglement wedge of $A$ highlighted in red falls short. }
    \label{exmp2}
\end{figure}
\\
\\
The minimum entanglement wedge reaches the true geometric wedge for the one layer pentagon-, single qubit- and pentagon/hexagon code for all boundaries. However, in general, may the true geometric geodesic exist between the tensors enclosed by $\gamma_{A}^{*}$ and those enclosed by $\gamma_{A^c}^{*}$. In that case $\gamma_A^{min}$ will not find the true geodesic entanglement wedge.
\subsection{Operator Reconstructability}
\label{reconstr}
In this section, the explicit reconstructability of operators acting on bulk qubits is investigated. Firstly, recall that the construction of the causal- and entanglement wedge was to ensure reconstructability of all operators. Since the tensor networks are made up of perfect tensors and the greedy algorithm will only advance past a perfect tensor when half or more tensor legs are included in the wedge, any operator may be reconstructed using the framework of perfect tensors introduced in Section \ref{massage}. As an example, in Figure \ref{case0} an arbitrary operator $O$ acting on the central bulk qubit of the one layer pentagon code with the boundary region visible in Figure \ref{exmp1} (left), is pushed through to the boundary subregion $A$. For this, all other boundary qubits and bulk qubits are considered as additional input legs to the five-qubit perfect tensor. Nonetheless, the entanglement algorithm ensures that at least half of the tensor legs are output. This ensures that the tensor will remain proportional to an isometric tensor. Consequently, the operator $\mathcal{O}$ = $O \otimes I \otimes I$ can be pushed through the tensor as $\mathcal{O} T = T T^{\dag} \mathcal{O} T $ = $ T O'$, where $O' = T^{\dag} \mathcal{O} T$. For the subsequent layer, the same framework can be applied to $O'$. Yet again identity operators are appended $O'$ until it is an operator working on all non-region legs. Only now the operator must be pushed through  $T \otimes T \otimes T$. This poses no problem since the greedy algorithm ensured that for each of the individual tensors T at least half of the legs are output. Therefore, T will still be proportional to an isometric tensor and the operator can be pushed through using the fact that ($T \otimes T \otimes T$)($T^{\dag} \otimes T^{\dag} \otimes T^{\dag}$) is proportional to the identity. In the end this procedure will result in an operator acting on solely the boundary subregion. In this example the resulting operator will be $O''$ = ($T^{\dag} \otimes T^{\dag} \otimes T^{\dag}$)(($T^{\dag}(O \otimes I \otimes I)T) \otimes I_3 \otimes I_3$) ($T \otimes T \otimes T$). \\
This procedure is not limited to the central bulk qubit and will work on all bulk qubits included in the greedy entanglement wedge at the same time. Since $I$, $X$, $Z$ and $Y$ from a basis for operators, these individual operators are able to be written as the tensor product of individual operators acting on all bulk qubits. As a result can these operators be pushed through by simply applying the operators on the bulk qubits and pushing them through at the correct moment. For example, in Figure \ref{case0}, operators acting on bulk qubits that are not the central qubit may be pushed through together with $O'$.
\begin{figure}[h!]
    \centering
    \includegraphics[scale = 0.70]{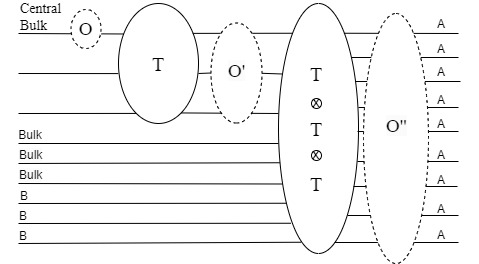}
    \caption{An example of an arbitrary operator $O$ acting on the central bulk qubit of the one layer pentagon code with the boundary visible in Figure \ref{exmp1} (left), being pushed through as the operator $O''$ acting on that boundary. }
    \label{case0}
\end{figure}
\subsubsection*{Minimum entanglement wedge}
As previously mentioned, the minimum entanglement wedge loses this assured operator reconstructability. If the minimum entanglement wedge of $A$ is equal the greedy entanglement wedge of $A$, there is no problem and all operators of included bulk qubits are reconstructable on $A$. The troubles comes when the minimum entanglement wedge is equal to $\gamma^*_{A^c}$. Because this geodesic generally does not come from $\gamma^*_A$, it is not assured that the formalism similar to Figure \ref{case0} is able to be applied. As an example, one is unable to use this same greedy formalism to send operators acting on the central qubit of Figure \ref{exmp2} to the boundary $A$. More precisely, this algorithm will fail by the third time an operator is send through as the amount of input legs becomes bigger then the output legs.
\subsubsection*{Geometric wedge}
While bulk qubits that are outside the greedy entanglement wedge are not guaranteed reconstructability on the boundary subregion, by exploiting symmetrical elements of the tensor network, some operators acting on bulk qubits beyond the greedy entanglement wedge can be reconstructed. While this has also been shown by \cite{Happy}, a more explicit example can be found by examining the network in Figure \ref{exmp2}. In this two layer single qubit code with a non-uniform boundary the greedy algorithm does not advance past the central qubit, but the area enclosed by true geometric geodesic, found by applying the minimum entanglement wedge, and the boundary subregion $A$ does enclose this central qubit. Conveniently, operators acting on this central qubit exists that are able to be reconstructed on $A$. To show this claim, inherent properties of the five-qubit tensor as a quantum error correcting code are exploited. Recall that this five-qubit code is a stabilizer code with its generators given in Table \ref{1}. These generators bring forth the stabilizer group $S$ of the five-qubit code under which the code subspace $P$ of logical qubits remains invariant. In other words, applying any element of the stabilizer group to the states which represent the logical qubits will simply rearrange the order of elements, but keep the states invariant. \\ \\
Consequently, as the code subspace remains invariant, the five-qubit tensor that encodes this code subspace to the correct logical qubit will also remain invariant under applying elements of the stabilizer. This is visible in Figure \ref{help}, where the five-qubit tensor remains invariant under the generator $M_1$ applied to its input legs. However, the legs of this tensors may be switched from input leg to output leg without a problem. This has been verified by explicitly determining the contracted tensor. Now, for the tensor network in Figure \ref{exmp2} there has been no explicit assigning of the tensor legs to corresponding qubits, apart from the bulk qubit in the central tensor. By cleverly choosing the boundary leg-qubit correspondence, replacing the tensors by tensors with operators according to Figure \ref{help} and changing the tensor legs as input and output accordingly, the outer layer of Figure \ref{case} can be constructed. As is visible in the figure, this arrangement will have operators acting only on boundary qubits of the subregion $A$. These operators get pushed trough to be Z operators that cancel with each other respectively and X operators acting on the central tensor legs. These X operators will then construct the logical $\overline{X}$ operator of the five-qubit code, which is pushed through to the central bulk qubit. This argument can be repeated in the opposite direction which allows the bulk X operator to be recreated on A.
\begin{figure}[h!]
    \centering
    \includegraphics[scale=0.50]{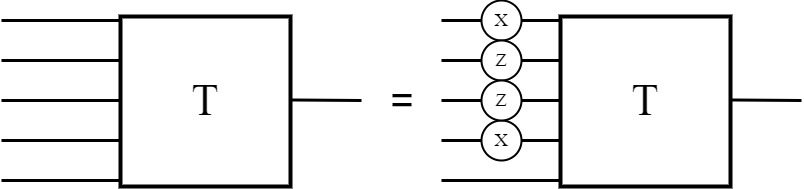}
    \caption{An example of how the five-qubit tensor remains invariant under the operator $M_{1}$ that is part of the stabilizer group of the five-qubit, visible in Table \ref{1}. }
    \label{help}
\end{figure}
\begin{figure}[h!]
    \centering
    \includegraphics[scale=0.58]{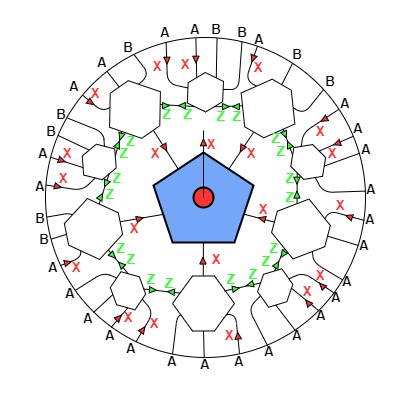}
    \caption{An example of how the X operator acting on the central qubit is recreated on a boundary $A$ for a tensor network where the central qubit is included in the geometric entanglement wedge but not the greedy entanglement wedge.}
    \label{case}
\end{figure}
\\ \\
The Z operator acting on the central bulk qubit can be reconstruced in much the same way. This can be done replacing all the X operators by Z operators in Figure \ref{case} and vice versa. While the X operators will cancel each other, the bulk Z operator will be reconstructed by Z operators on the boundary subregion $A$. However, it must be pointed out 
that for this to work with the generator $M_{1}$ the tensor legs/code qubit correspondence must be changed. \\ \\
While the tensor legs/code qubit correspondence was chosen according to our wishes, it has been done in such a way that the symmetries of the tensor network remain assuredly. It is unclear if allowing rationally invariant tensors to have the same legs corresponding to different code qubits changes something fundamentally of the tensor networks. As an example, it can be argued that this stabilizer method is limited if one chooses the legs that are connected to the central tensor to correspond the logical output of the connected hexagons. Because all stabilizer elements applied to this hexagon will always leave that leg unchanged, no non-trivial operator could be pushed through by this method. Consequently, while one might push operators to other legs, all legs connected the central tensor will remain the same. Thus, the central bulk qubit would remain unchanged as well. However, this limitation can be lifted when working with the six-qubit perfect state. For this state, stabilizers that have non trivial operators on the logical leg  exists. These are written down by Pastawski et al. in appendix A \cite{Happy}. By using these stabilizer elements, one is able to push logical operators through by the same method, regardless if the legs that are connected to the central tensor correspond to the logical output of the connected hexagons. This is an example that seems to indicate that the tensor legs/code qubit correspondence has no immediate influence on the global properties of the tensor network.
\newpage
\section{Mutual Information}
\label{3}
In this section the mutual information between a boundary subregion $A$ and the central qubit will be investigated. It will be examined if the inclusion of the central bulk qubit into the geometric wedge has any correlation with a rise in mutual information between the bulk qubit and boundary A. For this, the mutual information is given by, \begin{equation}
    I_{A,bulk} = S_{A} + S_{bulk} - S_{A,bulk}. 
\end{equation}
Because the sum of entropy of the individual systems will always be larger or equal to the entropy of the combination, the mutual information will be larger or equal to zero. The entropy for the toy tensor network models will be calculated using cut sizes, analogously to the corresponding holographic Formulas (\ref{entr3}, \ref{error corected RT}). However, one can also determine the entropy by calculating the Von Neumann entropy explicitly by tracing out the regions. A more in depth examination between these two different calculations has been performed by M. Grandjean \cite{max}. 
\subsubsection*{Classical RT Entropy}
When explicitly working with wedge size, one must define which algorithm was utilized. In this first section, no bulk contributions to the wedge size was allowed. As proven in \cite{Happy}, this wedge size will correspond to its holographic equivalent, the RT Formula (\ref{entr3}), for connected boundary subregions. Because of this, the mutual information will be investigated for a connected boundary subregion that is enlarged by one qubit at a time until the whole boundary region is reached. For one layer tensor networks, the minimum entanglement algorithm will be used to find the true geometric wedge. This allows in turn to determine the classical RT entropy. Since the bulk qubits are not chosen to contribute to wedge sizes, $S_{bulk}$ is always zero. The entropy of the boundary subregion together with the mutual information between the central bulk qubit and the moment of inclusion of that bulk qubit in geometric wedge, is plotted in Figure \ref{mut1} for the one layer pentagon- and single qubit code.
\begin{figure}[h!]
\minipage{0.50\textwidth}
  \includegraphics[width=\linewidth]{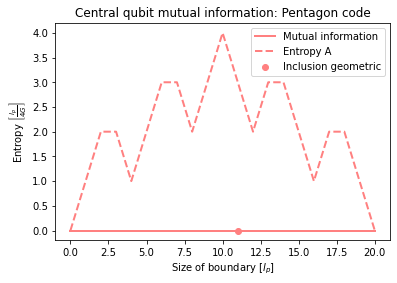}
\endminipage\hfill
\minipage{0.50\textwidth}%
  \includegraphics[width=\linewidth]{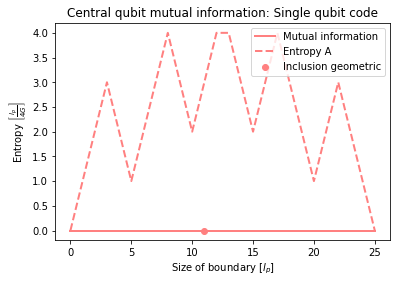}
\endminipage
\caption{The classic RT entropy of an increasing connected boundary region $A$ in function of the size of that boundary. In addition, the mutual information between central qubit and $A$  and the moment of inclusion in geometric wedge of the central bulk qubit. On the left the one layer holographic pentagon code and on the right the one layer single qubit code.}
\label{mut1}
\end{figure}
\\ 
In Figure \ref{mut1} it is visible that the mutual information remains zero for all boundary sizes. While this seems troublesome, it can be explained through the fact that the bulk qubits are not accounted for in the wedge size. In calculating the entropy $S_{A,bulk}$ the bulk central qubit was allowed to contribute to the greedy entanglement of $A$. However, since the minimum entanglement wedge was taken, the moment that the central bulk pushed the greedy geodesic $|\gamma_A^*|$ to include the central tensor, the compliment geodesic  $|\gamma_{A^c}^*|$ had the smaller size and kept the entropy equal to $S_A$. The entropy of both systems are also symmetric and has the property that a full boundary has an entropy of zero. This first property is an artifact of the starting point of the boundary and the inherent symmetry of the tensor network. The connected boundary region was started at an outer leg of a tensor and first moved to include that whole tensor. However, when starting from a middle tensor leg this symmetry fades. \\ \\ The behaviour of these entropy graphs can also be explained using the holographic state of the five-qubit code. The holographic state of the five-qubit tensors was given by $\frac{\ket{0 \overline{0}} + \ket{1 \overline{1}}}{\sqrt{2}}$. where the logical states where defined in Section \ref{try}. Not increasing the size of the wedge for the bulk qubits encapsulated, is proven to be numerically equivalent to assigning the collapsed states $\ket{0}$ or $\ket{1}$ to the bulk qubits. Consequently, the state of the entire system will change from the holographic state to be either $\ket{0 \overline{0}}$ or $\ket{1 \overline{1}}$. However, this projection forces the bulk qubit to lose the entanglement entropy as tracing out the noisy qubits will just yield $\ket{0}$ or $\ket{1}$ respectively. These have no entanglement entropy as they are pure states. Consequently, $I_{A,bulk}$ remains zero and the entropy graph ends on zero as well. \\ \\
In Figure \ref{mut2}, approximations of the RT formula are visible. In the left most figure random boundary qubits are gradually appended to the boundary. This is in contrast with the connected region of Figure \ref{mut1}. As a result the wedge size will only give an upper bound to the entropy according to Formula (\ref{entr 2}). It is visible on this graph that the upper bound will rise to 8, significantly higher than the connected region's peak, until about half the qubits are included in the boundary after which the upper bound on the entropy will decline until zero. \\ 
On the right graph of Figure \ref{mut2}, the RT entropy of an increasing boundary region for a two layer pentagon code is shown. However, since the pentagon code has two layers, it is no longer guaranteed that the minimal entanglement wedge will find the true geometric wedge. As a result the entropy will be an approximation of the true RT entropy. This is visible by the mutual information no longer remaining zero.
\begin{figure}[h!]
\minipage{0.50\textwidth}
  \includegraphics[width=\linewidth]{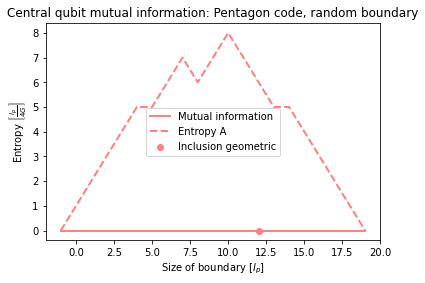}
\endminipage\hfill
\minipage{0.50\textwidth}%
  \includegraphics[width=\linewidth]{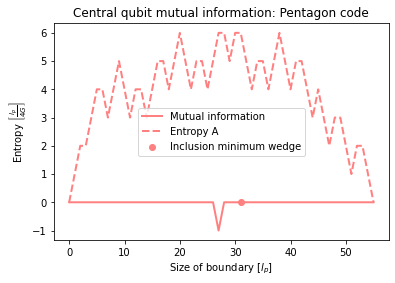}
\endminipage
\caption{Approximations of the RT entropy of an increasing boundary region $A$. In addition, the mutual information between central qubit and $A$  and the moment of inclusion in geometric wedge of the central bulk qubit. On the left a one layer holographic pentagon code with a random boundary and on the right a two layer pentagon code with connected boundary.}
\label{mut2}
\end{figure}
\subsubsection*{Error Corrected Entropy}
Similar to the second part of Section \ref{sim}, it can also be investigated how the entropy of the holographic state changes with boundary size. Since every bulk qubit will now be maximally entangled, the algorithm of the minimum wedge will, for boundary A, add one to the size of entanglement wedge $|\gamma_A^*|$ whenever a bulk qubit is encapsulated. Furthermore, for $A^c$ all bulk qubits encapsulated between $A$ and $\gamma_{A^c}^*$, after finishing the greedy algorithm on the compliment of the boundary, will add one to the size of $|\gamma_{A^c}^*|$. Consequently, the entropy of the boundary that may be calculated with this the wedge size will also change form its RT Formula (\ref{entr3}) equivalent to the Von Neumann equivalent of an error correcting network (\ref{error corected RT}) for which the $S_{bulk} (\rho_{\epsilon_A}) $ term is just the amount of bulk qubits encapsulated by the wedge and boundary. Additionally, the $S_{bulk}$ term of the central qubit mutual information is no longer zero but one. \\ \\
In Figure \ref{mut3} this entropy as a function of an increasing connected boundary region is visible for the one layer pentagon- and the one layer single qubit code. Additionally, the mutual information between the central qubit and the boundary region in function of increasing boundary size is plotted as well. Lastly, the moment of inclusion in the geometric wedge, determined by the minimum entanglement wedge algorithm, is also marked. Here it is visible that for both the pentagon code and the single qubit code, the entropy graphs are no longer symmetrical. Furthermore, the entropy of the whole boundary region will no longer be zero but the total amount of bulk qubits in the tensor network. This is expected as each bulk qubit is maximally entangled with the entire boundary through a series of five-qubit tensors. Most crucially, the mutual information between the central bulk qubit and the boundary no longer remains zero for this corrected entropy. For both the pentagon- and the single qubit code it is visible that mutual information remains at zero until over half of the boundary qubits are included in the region, which is when the boundary is larger than 10 and 12.5 for pentagon- and single qubit code respectively. At this moment, the following happens: firstly, the central bulk qubit is included in the minimum entanglement wedge of $A$. Consequently, the geodesic as a result of the minimum entanglement algorithm on $A$, which is the geometric geodesic $\gamma_A$, and the geometric geodesic as a result of $A$ and the central bulk, $\gamma_{A,bulk}$, will be the same from this point and forward. As the cut of the central bulk does not need to be added in $|\gamma_{A,bulk}|$, it will be smaller than $|\gamma_A|$ by one. This results in the observed mutual information of two.
\begin{figure}[h!]
\minipage{0.50\textwidth}
  \includegraphics[width=\linewidth]{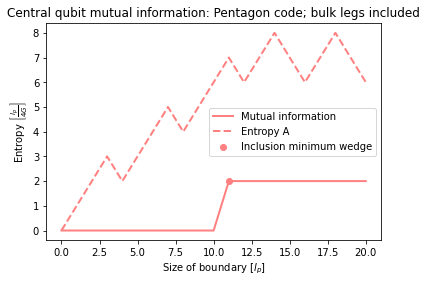}
\endminipage\hfill
\minipage{0.50\textwidth}%
  \includegraphics[width=\linewidth]{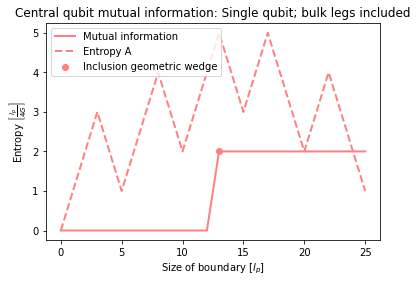}
\endminipage
\caption{The corrected entropy of an increasing connected boundary region $A$ in function of the size of that boundary. In addition, the mutual information between central qubit and $A$  and the moment of inclusion in geometric wedge of the central bulk qubit. On the left the one layer holographic pentagon code and on the right the one layer single qubit code.}
\label{mut3}
\end{figure}
When this happens, for both the pentagon- as the single qubit code, the true geometric wedge, equal to the minimum entanglement wedge, will include the central qubit.
\begin{figure}[h!]
\minipage{0.50\textwidth}
  \includegraphics[width=\linewidth]{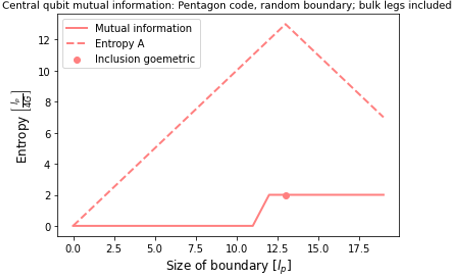}
\endminipage\hfill
\minipage{0.50\textwidth}%
  \includegraphics[width=\linewidth]{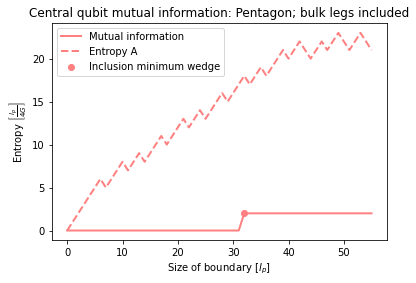}
\endminipage
\caption{Approximations of the corrected entropy of an increasing boundary region $A$. In addition, the mutual information between central qubit and $A$  and the moment of inclusion in geometric wedge of the central bulk qubit. On the left a one layer holographic pentagon code with a random boundary and on the right a two layer pentagon code with connected boundary.}
\label{mut4}
\end{figure}
\\
In Figure \ref{mut4}, similar to the previous section, the same graphs are visible for both random boundary addition and a two layer pentagon code. However, since a random boundary was used, the wedge size only provides an upper bound for the entropy according to Formula (\ref{entr 2}). Moreover, it is not guaranteed that the minimal entanglement wedge finds the true geometric wedge for the two layer pentagon code. Therefore the right most graph in Figure \ref{mut4} is also an approximation. \\ \\Nonetheless, nearly the same behaviour as Figure \ref{mut3} is visible in Figure \ref{mut4}. The graphs will fluctuate and end on the amount of bulk qubits in the system. The randomness on the boundary is visible in the left graph of Figure \ref{mut4}, as the entropy now moves in a straight line. This graphs also shows a different behaviour in mutual information with respect to geometric inclusion. While the mutual information will still rise to be two, it will no longer do it at half the boundary qubits. Moreover the inclusion in the geometric wedge will be even later. However, this mutual information is the result of an approximation of the RT formula, which does not guarantee accurate mutual information. In the next section, the entropy and mutual information for this particular tensor network with the same random boundary will be precisely determined with the explicit Von Neumann entropy using the holographic state. \\
The two layer pentagon code for a connected region does show the same behaviour. The mutual information will rise to two at the same time as over half the boundary qubits are included and the geometric wedge encapsulates the the central bulk qubit at the same time.
\subsubsection*{Mutual Information wedge}
In this last section the viability of the mutual information wedge is discussed. The mutual information wedge that is proposed can be defined as the wedge that includes all bulk qubits where the mutual information between the boundary and the bulk qubit is two.
This wedge has the benefit that it is dependent on the state in the bulk. This is visible in Figure \ref{mut1}, where the mutual information for the central bulk qubit is zero and therefore never included. From the examples of Figure \ref{mut3} and Figure \ref{mut4}, it does seem that there is some correlation between inclusion in the geometric wedge and mutual information. For this corrected entropy, the inclusion of greedy entanglement wedge is associated to rise in mutual information. However, this might be an artifact of the connected boundary subregions. \\
\\
The validity of this wedge will be investigated using the random boundary corresponding to the graph of Figure \ref{mut4} (left). Here, the mutual information will not be calculated using the wedge size but rather the Von Neumann entropy of the holographic state. This allows to determine the entropy precisely for both the random disconnected region $A$ and the union of $A$ and the central bulk qubit. Consequently, the mutual information will also be correct. \\ 
In Figure \ref{final}, this graph of the Von Neumann entropy and mutual information is visible. Because of the difficult nature of these computations, the entropy has only been determined up to a boundary size of twelve. Nonetheless, it is already visible that mutual information rises sooner than the inclusion in the geometric wedge. The rise in mutual information occurs at a boundary size of eight and corresponds for this example to the first boundary where three legs of one of the outer five-qubit tensors is included. However, an investigation of different examples proves that this correspondence does not occur in all examples. Additionally, while the RT upper bound on the entropy of $A$, visible in Figure \ref{mut4}, is reached, the rise in mutual information is sooner than the approximate graph. Therefore, it can be stated that entropy of the union of $A$ and the central bulk qubit does not follow its wedge upper bound.
\begin{figure}[h!]
    \centering
    \includegraphics[scale = 0.62]{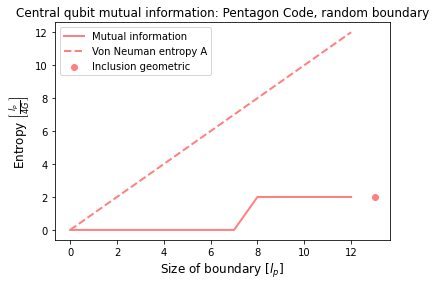}
    \caption{The Von Neumann entropy of an increasing random boundary region $A$ of a one layer pentagon code. In addition, the mutual information between central qubit and $A$  and the moment of inclusion in geometric wedge of the central bulk qubit.}
    \label{final}
\end{figure}
\\
Figure \ref{final} indicates that a rise in mutual information is not a valid indication for the inclusion in the geometric wedge. Moreover, the rise in mutual information is not even associated to the inclusion in the entanglement wedge. While this has been confirmed by explicitly working out the example, it is also visible in the graph directly since the entanglement wedge for a subregion of only eight qubits will never include the bulk qubit. As a result, a mutual information wedge which includes all bulk qubits with mutual information of two is not a valid candidate wedge to assure operator reconstructability. The associated rise in mutual information has only been found in an increasing connected boundary region where the boundary starts on the outer leg of a tensor and moves to first include that tensor. Lastly, no central qubit operator was reconstructable on that random boundary subregion until it was included in the geometric wedge.
\newpage
\section{Conclusion}
The tensor network framework has different wedges for boundary subregions that enclose the bulk qubits for which logical operators are reconstructable on these subregions. These were investigated by performing Monte-Carlo simulations of the boundary erasure and central qubit inclusion for the HaPPY five-qubit toy tensor networks \cite{Happy}. The simulations indicated that, while the causal wedge will not have any detectable threshold probability, the greedy entanglement wedge will have detectable threshold probabilities. The existence of these detectable threshold probabilities were strongly influenced by the structure of the tensor network. This is highlighted by the fact that a threshold probability was found for the single qubit code, but none was detected for the holographic pentagon code. Simulations of tensor networks with a non-uniform boundary showed its limited influence on a threshold probability. As shown by the stacking of tensor networks, the existence of such threshold probability seems to be an inherent property of the structure of the tensor network rather than the fraction $\frac{N_{bulk}}{N_{boundary}}$. Furthermore, the minimum entanglement wedge was introduced as the minimum of the greedy algorithms on the boundary and on the compliment. This showed good results in approximating the true geometric wedge, but lost its guaranteed operator reconstructability. An explicit example of operator reconstruction beyond the greedy entanglement wedge was found using the stabilizer elements of the five-qubit code. Lastly, the possibility of a wedge based on mutual information was rejected. Only for particular connected boundary subregions was the inclusion of the central tensor by the geometric wedge associated to a rise in mutual information.
\clearpage
\newpage
\section*{Acknowledgment}
I would like to thank my family and friends for their continued love and support. Particularly my mother, father and brother for keeping my spirits high and being always there for me. Additionally, I thank my brother-in-arms, Maxime, for the many insightful conversations and discussions.
I would also like to thank my supervisors, Juan and Philip, for without their guidance and wisdom this thesis would not have been possible. I thank my promotor Prof. Dr. Craps and the VUB for allowing me to do this research. 
I thank the authors of the HaPPY paper \cite{Happy}, for their clear and understandable language enabling me to recreate and expand on their research. Finally, I also thank the open source Python libraries \textit{Qiskit} and \textit{TensorNetworks} \cite{qiskit, tensornetwork}. Much of the code used can be found on \url{https://github.com/VicVanderLinden/AdS-CFT_HaPPY}.
\newpage
\addcontentsline{toc}{section}{
\end{document}